\pdfoutput=1
\documentclass[a4paper,UKenglish,cleveref, autoref, thm-restate]{lipics-v2021}
\usepackage{oursymbols}
\title{Communication Complexity of Private Simultaneous Quantum Messages Protocols}

\titlerunning{Communication Complexity of PSQM Protocols} %TODO optional, please use if title is longer than one line

\author{Akinori Kawachi}{Graduate School of Engineering, Mie University, Japan \and \url{https://sites.google.com/site/akinorikawachi/} }{kawachi@info.mie-u.ac.jp}{https://orcid.org/0000-0001-9218-9944}{JSPS Grant-in-Aid for Scientific Research (A) Nos.~16H01705, 21H04879, (B) No.~17H01695, JSPS Grant-in-Aid for Young Scientists (B) No.~17K12640, and MEXT Quantum Leap Flagship Program (MEXT Q-LEAP) Grant Number JPMXS0120319794.}%TODO mandatory, please use full name; only 1 author per \author macro; first two parameters are mandatory, other parameters can be empty. Please provide at least the name of the affiliation and the country. The full address is optional

\author{Harumichi Nishimura}{Graduate School of Informatics/Institute for Advanced Study, Nagoya University, Japan}{hnishimura@i.nagoya-u.ac.jp}{https://orcid.org/0000-0002-2219-3320}{JSPS Grant-in-Aid 
for Scientific Research (A) Nos.~16H01705, 21H04879, (B) No.~19H04066, Grant-in-Aid for Transformative Research Areas No.~20H05966 and MEXT Quantum Leap Flagship Program (MEXT Q-LEAP) Grant Number JPMXS0120319794.}

\authorrunning{A. Kawachi and H. Nishimura} %TODO mandatory. First: Use abbreviated first/middle names. Second (only in severe cases): Use first author plus 'et al.'

\Copyright{Akinori Kawachi and Harumichi Nishimura} %TODO mandatory, please use full first names. LIPIcs license is "CC-BY";  http://creativecommons.org/licenses/by/3.0/

\ccsdesc[500]{Theory of computation~Quantum computation theory}
\ccsdesc[500]{Theory of computation~Computational complexity and cryptography}
\keywords{Communication complexity, private simultaneous messages, quantum protocols, secure multi-party computation} %TODO mandatory; please add comma-separated list of keywords

\category{} %optional, e.g. invited paper

\relatedversion{} %optional, e.g. full version hosted on arXiv, HAL, or other respository/website
\acknowledgements{
We thank the anonymous reviewers of ITC 2021 for helpful comments.
}%optional

\nolinenumbers %uncomment to disable line numbering

\hideLIPIcs  %uncomment to remove references to LIPIcs series (logo, DOI, ...), e.g. when preparing a pre-final version to be uploaded to arXiv or another public repository

\EventEditors{Stefano Tessaro}
\EventNoEds{1}
\EventLongTitle{2nd Conference on Information-Theoretic Cryptography (ITC 2021)}
\EventShortTitle{ITC 2021}
\EventAcronym{ITC}
\EventYear{2021}
\EventDate{July 23--26, 2021}
\EventLocation{Bertinoro, Italy (Virtual Conference)}
\EventLogo{}
\SeriesVolume{199}
\ArticleNo{20}
\begin{document}

\maketitle

\begin{abstract}
The private simultaneous messages (PSM) model is a non-interactive version of the multiparty secure computation (MPC), which has been intensively studied to examine the communication cost of the secure computation. We consider its quantum counterpart, the {\em private simultaneous quantum messages (PSQM)} model, and examine the advantages of quantum communication and prior entanglement of this model.

In the PSQM model, $k$ parties $P_1,\ldots,P_k$ initially share a common random string (or entangled states in a stronger setting), and they have private classical inputs $x_1,\ldots, x_k$. Every $P_i$ generates a quantum message from the private input $x_i$ and the shared random string (entangled states), and then sends it to the referee $R$.
Receiving the messages from the $k$ parties, $R$ computes $F(x_1,\ldots,x_k)$ from the messages. Then, $R$ learns nothing except for $F(x_1,\ldots,x_k)$ as the privacy condition.

We obtain the following results for this PSQM model. ($i$) We demonstrate that the privacy condition inevitably increases the communication cost in the two-party exact PSQM model as well as in the classical case presented by Applebaum, Holenstein, Mishra, and Shayevitz [{\em Journal of Cryptology} {\bf 33}(3), 917--953 (2020)]. In particular, we prove a lower bound $(3-o(1))n$ of the communication complexity in two-party exact PSQM protocols with a shared random string for random Boolean functions of $2n$-bit input, which is larger than the trivial upper bound $2n$ of the communication complexity without the privacy condition. ($ii$) We demonstrate a factor two gap between the communication complexity of exact PSQM protocols with shared entangled states and with shared random strings by designing a multiparty PSQM protocol with shared entangled states for a total function that extends the two-party equality function. ($iii$) We demonstrate an exponential gap between the communication complexity of exact PSQM protocols with shared entangled states and with shared random strings for a two-party {\em partial} function.
\end{abstract}

\newpage

\section{Introduction}

{\bf Background.} 
Communication complexity has been an important research area in theoretical computer science for more than four decades, aiming to understand the communication cost of computing functions in a distributed manner~\cite{Yao79,KN97book}.
Since the advent of quantum information science, quantum communication complexity has also been studied intensively 
to determine the advantage of quantum information processing over its classical counterparts. 
A number of studies have succeeded in demonstrating the quantum advantages from the early days of quantum complexity theory~\cite{BCW98STOC,Raz99STOC,BCWW01PRL}.

Recently, much attention has been given to studying the amount of communication overhead required to preserve privacy in the field of cryptography, particularly, multi-party secure computation (MPC), to explore the optimal communication cost for privacy from the viewpoint of communication complexity~\cite{DPP16,DLN19}. 
MPC is commonly based on a general network model that has complex communication patterns (e.g., in which each of many parties can freely interact with the other parties bidirectionally) unlike standard models in communication complexity (e.g., in which two parties can exchange messages only with each other). Therefore, many studies have focused on a special class of MPC that has simpler communication patterns, such as {\em private simultaneous messages}\ (PSM) protocols~\cite{FKN94,IK97,BIKK14,BKN18,AHMS20JCRYPT}.

The two-party version of the PSM model was first proposed by Feige, Kilian, and Naor~\cite{FKN94}, and was later extended by Ishai and Kushilevitz~\cite{IK97}.
In the general setting of the PSM model, we consider $k$ parties $P_1,\ldots,P_k$ and a unique referee $R$.
The party $P_i$ has its private input $x_i$, and all parties share a common random string $r$. Each $P_i$ generates message $m_i$ from $x_i$ and $r$, and then, sends $m_i$ to the referee $R$ only once. Note that each party is not allowed to interact with other parties.
The referee $R$ receives the messages $m_1,\ldots,m_k$, and computes an output value of a predetermined function $F$.
The protocol generally has two properties {\em correctness}\ and {\em privacy}: 
Correctness signifies that the referee can compute $F(x_1, \ldots, x_k)$ correctly from the messages $m_1, \ldots, m_k$, 
while privacy signifies that the referee $R$ learns nothing except for $F(x_1, \ldots, x_k)$ from the received messages $m_1, \ldots, m_k$ in the information-theoretical sense.

In fact, the communication model of PSM protocols coincides with simultaneous message passing (SMP) protocols, 
which are known as traditional communication models in communication complexity~\cite{Yao79,KN97book}.
In the (number-in-hand) SMP model, $k$ parties $P_1, \ldots, P_k$ that share a common random string $r$ (and sometimes entangled states), send their messages $m_1, \ldots, m_k$ computed from individual inputs $x_1, \ldots, x_k$ and the referee computes $F(x_1, \ldots, x_k)$ from $m_1, \ldots, m_k$, as performed in the PSM model. Note that the SMP model does not require the privacy condition unlike the PSM model. The communication complexity of SMP protocols has been widely studied from the viewpoint of classical/quantum information to demonstrate the power of quantum communication~\cite{BCWW01PRL,GKRW09SICOMP,Gav20}. 

Two-party quantum SMP models were first studied by Buhrman, Cleve, Watrous, and de Wolf~\cite{BCWW01PRL} in the setting that the two parties do not share any randomness or entanglement. %% also the two parties send quantum messages to referee
In this model, they demonstrated that the quantum communication complexity of the equality function is exponentially smaller than in the classical case. This result has been strengthened in the literature~\cite{BYJK08,Gav19}, and Gavinsky~\cite{Gav20} demonstrated that there is a relational problem whose quantum communication complexity is exponentially smaller than that of the two-way classical communication model.
However, the power of shared entanglement in the SMP model is unclear. In one of the few related studies, Gavinsky, Kempe, Regev, and de Wolf~\cite{GKRW09SICOMP} demonstrated that there is a relational problem that has an exponential gap between quantum SMP models with shared entanglement and without shared entanglement. However, the known maximum gap between them for {\em total functions} is only a constant multiplicative factor of $2$~\cite{HSW+05,KK12}. 

Although various studies have examined quantum versions of MPC so far (e.g., \cite{CGS02,Unr10,DGJMS20}), to the best of the authors' knowledge, there has been no attempt to analyze quantum communication complexity under the privacy condition in a cryptographic setting, and such analysis is important to understand the advantages of quantum communication in a cryptographic setting.

{\bf Contributions.} 
In this paper, we examine the power of quantum communication and shared entanglement under the information-theoretical privacy condition based on a standard communication model, namely, the PSM (or, equivalently, SMP) model. In particular, we propose a quantum counterpart of the classical PSM model called {\em private simultaneous quantum messages (PSQM)} model.
In the PSQM model, parties $P_1,\ldots,P_k$ which have classical private inputs $x_1,\ldots,x_k$ share a common random string or entangled states in advance, and can send quantum messages to a quantum referee, $R$. Then, $R$ computes a classical output value $F(x_1,\ldots,x_k)$ for a given function $F$.

In the PSM (and its related) model, there are few results on lower bounds 
of communication complexity \cite{AHMS20JCRYPT,DPP16,BHILM20}. 
As one of such results, Applebaum, Holenstein, Mishra, and Shayevitz~\cite{AHMS20JCRYPT} proved 
a lower bound $(3-o(1))n$ of the communication complexity for random functions $F_n: \bit^n\times\bit^n\rightarrow\bit$ in the PSM model. In contrast, 
every function has the trivial upper bound $2n$ in the SMP model 
(i.e., the PSM model without the privacy condition). 
This result implies that the privacy condition creates communication overhead in the PSM model for some functions.
Our first result demonstrates that this communication overhead is inevitable even if the parties can send quantum messages as in the PSQM model.
\begin{theorem}\label{thm:lowerbound}
For a $(1-o(1))$ fraction of functions $F_n:\bit^n\times\bit^n\rightarrow\bit$,
the communication complexity of two-party exact PSQM protocols with shared randomness for $F_n$ is at least $3n-2\log{n}-O(1)$.
\end{theorem}

We also present a multiparty PSQM protocol for a total function that reduces the amount of quantum communication by half under the condition that the parties share entanglement compared to the case in which they do not share entanglement. 

\begin{theorem}\label{thm:upperbound}
For any even $n$ and $k$, there is a total function $F_n:(\{0,1\}^n)^k\rightarrow\{0,1\}$ such that the communication complexity of the $k$-party exact PSQM protocol with shared entanglement is at most $kn/2$, while that without shared entanglement is $kn$. 
\end{theorem}

Actually, this function matches the equality function for the two-party case. It is known that for the equality function, the two-party quantum SMP model with shared entanglement reduces the amount of quantum communication by half compared to the corresponding model without shared entanglement (e.g.~\cite{HSW+05}). Our result implies that this reduction still holds even if the privacy condition is required.  

Moreover, we present a two-party PSQM protocol with shared entanglement for a {\em partial} function that reduces the amount of quantum communication exponentially compared to the case in which the parties do not share entanglement. 

\begin{theorem}\label{thm:exp-gap-DJ}
There is a partial function $F_n:\{0,1\}^n\times\{0,1\}^n\rightarrow\{0,1\}$ 
such that the communication complexity of the exact PSQM protocol with shared entanglement is $O(\log n)$ while that without entanglement is $\Omega(n)$.
\end{theorem}

{\bf Related Work.} 
% quantum communication complexity with privacy
There have been several studies on quantum communication complexity 
with privacy conditions~(e.g.,~\cite{Kla04,GI13}), 
although they differed from a cryptographic setting. 
For example, Gavinsky and Ito~\cite{GI13} considered the SMP model with privacy; however it considered the information leakage of quantum messages when the input was randomly chosen, 
while in the cryptographic setting, privacy should be retained 
for any input.

% quantum randomized encoding scheme
In a study related to the PSQM model, Brakerski and Yuen~\cite{BY20} constructed a quantum version of decomposable randomized encoding schemes. 
In fact, decomposable randomized encoding is equivalent to the PSM model from a communication-complexity perspective.
They demonstrated how to garble a general quantum circuit on quantum inputs in a decomposable manner via a constant-depth quantum circuit. In contrast, our study focuses on the communication complexity of computing several classical functions on classical inputs in the communication model.

More recently, Morimae~\cite{Mor20} investigated relationships between quantum randomized encoding and other quantum protocols including quantum computing verification and blind quantum computing. For example, he proved that a randomized encoding scheme of the BB84 state generation implies a two-round verification scheme of quantum computing with a classical verifier that additionally performs the encoding operation, and that a quantum randomized encoding scheme with a classical encoding operation implies violation of the no-cloning theorem. His target of quantum randomized encoding schemes is similar to that of~\cite{BY20}, 
that is, encoding for quantum circuits on quantum inputs rather than classical functions on classical inputs.

\section{Preliminaries}

Let $[n]:=\{1,2,\ldots,n\}$. All logarithms are to base two.
For any two $m$-bit strings $x=x_1\cdots x_m$ and $y=y_1\cdots y_m$, the product $x\cdot y$ denotes $\sum_{i\in [m]} x_iy_i~(\mbox{mod}~2)$, 
and $x\oplus y$ denotes the $m$-bit string 
whose $i$th bit is the XOR of $x_i$ and $y_i$.

For any $m$-bit string $x=x_1x_2\cdots x_m$, 
let 
\begin{equation}\label{eq:polynomial-rep}
{\sf p}(x)=x_1+x_2\alpha+\cdots+x_m\alpha^{m-1}~(\mbox{mod}~{\sf q}_m)
\end{equation}
be the corresponding polynomial over $\mathbb{F}_2$, 
where ${\sf q}_m$ is some irreducible polynomial of degree $m$ over $\mathbb{F}_2$. 
Note that ${\sf p}$ is an $\mathbb{F}_2$-linear one-to-one correspondence
between $\{0,1\}^m$ and $\mathbb{F}_{2^m}$ with ${\sf p}(0^m)=0$ and
${\sf p}(x\oplus y)={\sf p}(x)+{\sf p}(y)$; in particular, it maps
$\{0,1\}^m\setminus\{0^m\}$ onto the multiplicative group $\mathbb{F}^*_{2^m}$.

We assume the reader is familiar with the basics of quantum information and computation such as quantum states and quantum operations (see, e.g, \cite{HIK+15,NC00}). According to the standard notations, Pauli gates $X$, $Z$, and the Hadamard gate $H$ denote
\[
X=
\left(
\begin{matrix}
0 & 1\\
1 & 0
\end{matrix}
\right),\ \ Z=
\left(
\begin{matrix}
1 & 0\\
0 & -1
\end{matrix}
\right),\ \ H=
\frac{1}{\sqrt{2}}
\left(
\begin{matrix}
1 & 1\\
1 & -1
\end{matrix}
\right),
\]
respectively.

\subsection{Private simultaneous quantum messages protocols}

Private simultaneous quantum messages (PSQM) protocols are formally defined as follows.

\begin{definition}[private simultaneous quantum messages (PSQM) protocols]
For positive integers $n,k>0$, 
let $n$ be the size parameter and $k$ be the number of parties. 
Let $F_n:\prod_{i=1}^k\XX_{n,i}\rightarrow\YY_n$.
We say that a $(k+1)$-tuple $\Pi=(P_1,\ldots,P_k,R)$ of quantum algorithms is an $\varepsilon$-error $k$-party private simultaneous quantum messages (PSQM) protocol 
if the following holds: given an individual input $x_i\in\XX_{n,i}$ and shared random string $r$ among $P_1,\ldots,P_k$, 
the $i$th party $P_i$ prepares a quantum message, represented as $\rho_i(x_i;r):=P_i(x_i;r)$, 
in a Hilbert space $\MM_{n,i}$ 
called a quantum register, %(where $\MM_{n,i}$ is the Hilbert space of quantum messages from $P_i$) 
and sends $\MM_{n,i}$ (or equivalently $\rho_i(x_i;r)$) to the party $R$, which is called the referee. 
Then, the following two properties hold: 
\begin{enumerate}
\item {\bf (Correctness)} 
The referee $R$ outputs the classical value $F_n(x_1,\ldots,x_k)\in\YY_n$ using the received joint quantum register $\MM_n:=\bigotimes_{i=1}^k\MM_{n,i}$ 
with a probability of at least $1-\varepsilon$. 
\item {\bf ((Perfect) Privacy)} 
There exists a quantum algorithm $S_n$, 
which is called the simulator, such that the output quantum state $S_n(F_n(x_1,\ldots,x_k))$ is identical to 
the quantum state in $\MM_n$ (before $R$), 
namely, $\sum_r\pi(r)\otimes_{i=1}^k \rho_i(x_i;r)$, where $\pi$ is a distribution of shared random strings.
\end{enumerate}
We say that the protocol is exact when $\varepsilon=0$.

If the shared random string $r$ is replaced by a predetermined multipartite entangled quantum state $\ket{\Phi}$ among the $k$ parties, 
we say that $\Pi$ is a PSQM protocol with a shared entangled state $\ket{\Phi}$, 
where the algorithms and the properties are similarly defined except that $P_i$ prepares the message using its own part of $\ket{\Phi}$ (instead of $r$), 
and that the quantum state in $\MM_n$ need not be a product state of the $k$ local states in $\MM_{n,1},\ldots,\MM_{n,k}$ any more. 

The communication complexity of $\Pi$ is defined by the total length 
$\log{\dim(\MM_n})$ of the messages.
%, where $\MM_n$ is the joint space of the quantum messages of $k$ players, namely, $\MM_n:=\bigotimes_{i=1}^k\MM_{n,i}$. 

\end{definition}

Let $C^{psm}_{\varepsilon}(F_n)$ (resp.~$Q^{psm}_{\varepsilon}(F_n)$) be the $\varepsilon$-error classical (quantum) communication complexity of the problem $F_n$ in the PSM (PSQM) model with a shared random string.
Let $C^{psm,*}_{\varepsilon}(F_n)$ (resp.~$Q^{psm,*}_{\varepsilon}(F_n)$) be the $\varepsilon$-error classical (quantum) communication complexity of $F_n$ in the PSM (PSQM) model with shared entangled states (the PSM model with shared entangled states is defined similarly to the PSQM model with shared entangled states except that the messages sent to the referee are restricted to classical strings).

\section{Communication Lower Bounds of Two-Party PSQM Protocols}
\newcommand{\eventP}{\PP^{(=)}}
\newcommand{\eventI}{\II^{(=)}}
\newcommand{\eventnotI}{\II^{(\ne)}}
\newcommand{\eventF}{\FF^{(=)}}
\newcommand{\eventFy}{\FF^{(=y)}}
\newcommand{\eps}{\varepsilon}
\newcommand{\ipHS}[2]{\langle #1,#2 \rangle_{\mathrm{HS}}}

In this section, we present the communication lower bounds 
of random functions for two-party PSQM protocols (Theorem~\ref{thm:lowerbound}).

%\ignore{%short paper用証明概略？
The proof strategy is based on that of the classical case presented by Applebaum et al.~\cite{AHMS20JCRYPT}.
The proof for the classical case considers 
two independent executions of a PSM protocol.
It then evaluated the upper bounds of the collision probability, that is, the probability that the message in the first execution 
coincides with the one in the second execution, between two independent random messages. 
Because the collision probability is lower-bounded by 
the inverse of the size of the message domain, 
we can obtain the communication lower bound from the upper bound 
of the collision probability. 
Note that this argument is not available for quantum messages 
since they vary infinitely even over a finite number of qubits.

In order to extend the above argument to the case of quantum messages, we use the {\em purity}, 
$\tr\rho^2$, of a quantum message $\rho$ in a PSQM protocol instead of the collision probability. 
In accordance with its name, the purity is originally a measure of how pure a quantum state is.
(For example, any pure state has a purity of $1$, and the $d$-dimensional maximally mixed state 
has $1/d$.)
It is easy to see that the purity of a quantum state $\rho$ is lower-bounded by $1/\dim(\rho)$, 
and thus, we can obtain the communication lower bounds for a PSQM protocol 
by evaluating the upper bound of the purity of 
the quantum messages, similarly to the collision probability for a PSM protocol.

However, the purity of quantum messages is different 
from the collision probability between classical messages; thus, we must further adapt the proof technique in \cite{AHMS20JCRYPT} to the purity. For example, while the collision probability is analyzed by combinatorial techniques in the proof of \cite{AHMS20JCRYPT}, 
we need to analyze the trace $\tr\rho^2$ combinatorially by extending the original proof (Claim~\ref{clm:purity upper bound}). 
Also, the proof technique in \cite{AHMS20JCRYPT} uses 
a unique collision (which is obtained from the property called non-degeneracy that random functions have with high probability) between two messages in two independent executions with any fixed shared random string.
Instead of the unique collision, we consider weighted collisions defined from 
the inner product of two quantum messages and extend the original argument for the weighted collisions 
(Lemma~\ref{lem:weight}).
%}

Before discussing the details of the proof, 
we provide several technical definitions and notation required for the proof of the lower bounds. 
In this section, we denote ${\cal X}_{n,i}$ by ${\cal X}_{i}$. 
We use $\rho(x_1,x_2;r)=\rho_1(x_1;r)\otimes\rho_2(x_2;r)$ to denote the entire quantum message sent from $P_1$ and $P_2$ on individual inputs $x_1\in{\cal X}_1$ and $x_2\in{\cal X}_2$ with a shared random string $r$ to $R$, where $\rho_1(x_1;r)$ denotes $P_1$'s message and $\rho_2(x_2;r)$ denotes $P_2$'s message. 
%Then, we set $\rho_1(x_1):=\sum_r \pi(r)\rho_1(x_1;r)$, $\rho_2(x_2):=\sum_r \pi(r)\rho_2(x_2;r)$, and $\rho(x_1,x_2):=\sum_r \pi(r)\rho(x_1,x_2;r)$, where $\pi(\cdot)$ is a uniform distribution over the domain of shared random strings.

Let $\mu$ be a distribution over $\XX_1\times\XX_2$ with marginal distributions $\mu_1$ and $\mu_2$.
We define $\Supp(\mu)$ for a distribution $\mu$ as a set $\set{x: \prob{X\sim\mu}{X=x}>0}$.
We say that function $F_n$ is non-degenerate under distribution $\mu$ 
if for every distinct $x_1\in\Supp(\mu_1)$ and $x^\prime_1\in\Supp(\mu_1)$, 
there exists $x_2\in\Supp(\mu_2)$ such that $F_n(x_1,x_2)\ne F_n(x^\prime_1,x_2)$ 
and for every distinct $x_2\in\Supp(\mu_2)$ and $x^\prime_2\in\Supp(\mu_2)$ 
there exists $x_1$ such that $F_n(x_1,x_2)\ne F_n(x_1,x^\prime_2)$.
We say that $F_n$ is non-degenerate if the above holds when replacing 
$\Supp(\mu_1)$ and $\Supp(\mu_2)$ by $\XX_1$ and $\XX_2$, respectively.

A rectangle $\RR$ of size $k\times\ell$ over $\XX_1\times\XX_2$ is defined as $((x_{1,1},\ldots,x_{1,k}),(x_{2,1},\ldots,x_{2,\ell}))$, where 
 $x_{1,i}\in\XX_1,x_{2,j}\in\XX_2$ for every $i,j$,  $x_{1,i}\ne x_{1,i^\prime}$ for every distinct $i, i^\prime$,
 and $x_{2,j}\ne x_{2,j^\prime}$ for every distinct $j,j^\prime$. 
We say that two rectangles 
%$\RR=((x_{1,i},x_{2,j}))_{i\in[k],j\in[\ell]}$
$\RR=((x_{1,1},\ldots,x_{1,k}),(x_{2,1},\ldots,x_{2,\ell}))$ and 
%$\RR^\prime=((x^{\prime}_{1,i},x^{\prime}_{2,j}))_{i\in[k],j\in[\ell]}$
$\RR^\prime=((x^{\prime}_{1,1},\ldots,x^{\prime}_{1,k}),(x^{\prime}_{2,1},\ldots,x^{\prime}_{2,\ell}))$
are $\XX_1$-disjoint (resp.~$\XX_2$-disjoint)
if $x_{1,i}\ne x^{\prime}_{1,i}$ for every $i\in[k]$ (resp.~if $x_{2,j}\ne x^{\prime}_{2,j}$ for every $j\in[\ell]$). In particular, we say that $\RR$ and $\RR^\prime$ are disjoint if they are either $\XX_1$-disjoint or $\XX_2$-disjoint.

For a rectangle  %$\RR=((x_{1,i},x_{2,j}))_{i\in[k],j\in[\ell]}$,
$\RR=((x_{1,1},\ldots,x_{1,k}),(x_{2,1},\ldots,x_{2,\ell}))$, 
let $F_n[\RR]$ be a matrix whose $(i,j)$-entry is $F_n(x_{1,i},x_{2,j})$, and 
let $\mu(\RR) = \sum_{i\in[k],j\in[\ell]} \mu(x_{1,i},x_{2,j})$.  
%$\mu(\RR) = \sum_{(x_{1},x_{2})\in\RR}\mu((x_{1},x_{2}))$.
We say that $\RR$ is similar to $\RR^\prime$ if $F_n[\RR]=F_n[\RR^\prime]$.
We define
\[
 \alpha(\mu):= \max_{(\RR_1,\RR_2)}\min\{\mu(\RR_1),\mu(\RR_2)\},
\]
where the maximum ranges over all pairs of similar disjoint rectangles $(\RR_1, \RR_2)$.
In addition,
\[
 \beta(\mu):= \min_y \CProb{(X_1,X_2),(X_1',X_2')\sim\mu}{(X_1,X_2)\ne(X_1',X_2')}{F_n(X_1,X_2)=F_n(X_1',X_2')=y},
\]
where $(X_1,X_2)$ and $(X_1',X_2')$ are independent, and $y$ is taken such that $\Pr_{(X_1,X_2)\sim\mu}[F_n(X_1,X_2)=y]>0$.

We can demonstrate the communication lower bound in Theorem~\ref{thm:lowerbound}
from the following main technical lemma combined with an appropriate function $F_n$, 
which is provided in the study by Applebaum et al.~\cite{AHMS20JCRYPT}.

\begin{lemma}\label{lem:LBmain}
For every non-degenerate function $F_n:{\cal X}_1\times {\cal X}_2\rightarrow\bit$, 
we have 
\[
Q^{psm}_0(F_n)\ge 
 \max_\mu\left(\log(\alpha(\mu)^{-1}) + H_\infty(\mu) - \log(\beta(\mu)^{-1})\right)-1,
\]
where $\mu$ is taken over all distributions over $\XX_1\times\XX_2$ 
under which $F_n$ is non-degenerate and for which $\beta(\mu)>0$,
and $H_\infty(\mu)$ is the min-entropy of $\mu$.
\end{lemma}

From the previous study \cite{AHMS20JCRYPT}, we can obtain the appropriate function 
by selecting a function at random, as illustrated in the following theorem.
\begin{theorem}[Applebaum et al.~\cite{AHMS20JCRYPT}]\label{thm:randomfunction}
For a $(1-o(1))$ fraction of the functions $F_n:\bit^n\times\bit^n\rightarrow\bit$, $F_n$ is non-degenerate and $\setsize{\RR}\le 2^n\cdot n^2$ holds for every pair $(\RR,\RR^\prime)$ of similar disjoint rectangles.
\end{theorem}

Consider the uniform distribution $U$ over $\{0,1\}^n\times\{0,1\}^n$.
We can show $\beta(U)\ge 1/2>0$ if $|\XX_1|,|\XX_2|\ge 3$ for a non-degenerate function $F_n$, as proved in the proof of Theorem 7 in \cite{AHMS20JCRYPT}.
The communication lower bound of PSQM protocols for $F_n:\bit^n\times\bit^n\rightarrow\bit$ is bounded by $\log(\alpha(U)^{-1}) + H_\infty(U) - \log(\beta(U)^{-1})-1$ from Lemma~\ref{lem:LBmain}. By Theorem~\ref{thm:randomfunction}, we can easily see that this bound is $3n-2\log{n}-O(1)$ for a $(1-o(1))$ fraction of the functions $\bit^n\times\bit^n\rightarrow\bit$ by the same argument used in \cite{AHMS20JCRYPT}, as given in Theorem~\ref{thm:lowerbound}.

Now, we provide the proof of the main technical lemma.
\begin{proof}[Proof of Lemma~\ref{lem:LBmain}]
From the correctness, the referee $R$ outputs $F_n(x_1,x_2)$ for the received quantum message $\rho(x_1,x_2;r)$ 
for every $x_1,x_2$ and every $r$. 
Without loss of generality, we can assume that 
$P_1$ and $P_2$ generate pure states $\ket{\psi_1(x_1;r)}$ and $\ket{\psi_2(x_2;r)}$ with a shared randomness $r$, respectively. 

This assumption is justified as follows.
Suppose that $P_1$ and $P_2$ generate mixed states $\rho_1(x_1; r)$
and $\rho_2(x_2; r)$ as their messages on inputs $x_1,x_2$ and a shared random string $r$.
By the spectral decomposition, we have 
$
 \rho_1(x_1;r) = \sum_{i} p_i(x_1;r) \ketbra{\psi_1^{(i)}(x_1;r)}{\psi_1^{(i)}(x_1;r)}
$ and 
$
 \rho_2(x_2;r) = \sum_{j} q_j(x_2;r) \ketbra{\psi_2^{(j)}(x_2;r)}{\psi_2^{(j)}(x_2;r)}.
$
The joint message state is
\[
 \rho(x_1,x_2;r) := \sum_{i,j} p_i(x_1;r) q_j(x_2;r) \ketbra{\psi_1^{(i)}(x_1;r), \psi_2^{(j)}(x_2;r)}{\psi_1^{(i)}(x_1;r), \psi_2^{(j)}(x_2;r)}
\]
and its probabilistic mixture over the shared random string $r$ following a distribution $\pi$ is 
\begin{align*}
 \rho(x_1,x_2) & := \sum_{r} \pi(r) \rho(x_1,x_2;r)\\
 & = \sum_{r,i,j} \pi(r) p_i(x_1;r) q_j(x_2;r) \ketbra{\psi_1^{(i)}(x_1;r), \psi_2^{(j)}(x_2;r)}{\psi_1^{(i)}(x_1;r), \psi_2^{(j)}(x_2;r)}.
\end{align*}
To obtain a mixing distribution independent of the actual inputs,
consider tables $\mathbf{i}=(i_u)_{u\in\XX_1}$ and
$\mathbf{j}=(j_v)_{v\in\XX_2}$ specifying an eigenstate index for every
possible local input. Define $\widetilde r=(r,\mathbf{i},\mathbf{j})$ with
\[
 \widetilde\pi(r,\mathbf{i},\mathbf{j})
 :=\pi(r)\prod_{u\in\XX_1}p_{i_u}(u;r)
             \prod_{v\in\XX_2}q_{j_v}(v;r).
\]
This is a normalized distribution independent of $(x_1,x_2)$.
For the mixture component $\widetilde r$, we use local pure states
\[
 \ket{\widetilde\psi_1(x_1;\widetilde r)}
 :=\ket{\psi_1^{(i_{x_1})}(x_1;r)},\qquad
 \ket{\widetilde\psi_2(x_2;\widetilde r)}
 :=\ket{\psi_2^{(j_{x_2})}(x_2;r)}.
\]
Marginalizing the unused table entries yields the weights
$\pi(r)p_i(x_1;r)q_j(x_2;r)$, namely, for any fixed $i,j$,
\begin{align*}
 \sum_{\substack{\mathbf{i},\mathbf{j}:\\ i_{x_1}=i, j_{x_2}=j}} \widetilde\pi(r,\mathbf{i},\mathbf{j}) & =\pi(r)p_i(x_1;r)q_j(x_2;r)\prod_{u\ne x_1}\left(\sum_{i_u}p_{i_u}(u;r)\right)\prod_{v\ne x_2}\left(\sum_{j_v}q_{j_v}(v;r)\right)\\
 & =\pi(r)p_i(x_1;r)q_j(x_2;r),
\end{align*}
where $\sum_{i_u}p_{i_u}(u;r)=\sum_{j_v}q_{j_v}(v;r)=1$ from the spectral decomposition,
and thus, averaging these pure product states
reproduces $\rho(x_1,x_2)$ for every input pair since we can write
\[
 \rho(x_1,x_2)=\sum_{r,i,j} \sum_{\substack{\mathbf{i},\mathbf{j}:\\ i_{x_1}=i, j_{x_2}=j}}\widetilde\pi(r,\mathbf{i},\mathbf{j})\ketbra{\widetilde\psi_1(x_1;\widetilde r),\widetilde\psi_2(x_2;\widetilde r)}{\widetilde\psi_1(x_1;\widetilde r),\widetilde\psi_2(x_2;\widetilde r)}.
\]
Hence, privacy and the
message-space dimensions are unchanged. Perfect correctness holds for
every positive-weight component since a convex combination of
nonnegative error probabilities can be zero only if each
positive-weight term has zero error probability.
Relabeling $\widetilde r$, $\widetilde\pi$, and $\widetilde\psi_i$ as $r$, $\pi$, and $\psi_i$ respectively, 
we suppose that their quantum messages can be represented as pure states without loss of generality. 
In the following, $r$ ranges over $\Supp(\pi)$.

We denote by $\ket{\psi(x_1,x_2;r)}$ their joint state $\ket{\psi_1(x_1;r)}\otimes\ket{\psi_2(x_2;r)}$.
Namely, we set 
\begin{equation}\label{eq:rho(x_1,x_2;r)}
    \rho(x_1,x_2;r):=\ketbra{\psi(x_1,x_2;r)}{\psi(x_1,x_2;r)}
    =\ketbra{\psi_1(x_1;r)}{\psi_1(x_1;r)}\otimes\ketbra{\psi_2(x_2;r)}{\psi_2(x_2;r)}.
\end{equation}
In addition, we set 
\begin{equation}\label{eq:rho(x_1,x_2)}
    \rho(x_1,x_2):= \sum_r\pi(r)\rho(x_1,x_2;r),
    %= \sum_{r}\pi(r)\ketbra{\psi(x_1,x_2;r)}{\psi(x_1,x_2;r),
\end{equation} 
where $\pi(r)$ denotes the probability that $r$ is selected as a shared random string.

\ignore{
As performed by Applebaum et al.~\cite{AHMS20JCRYPT}, we consider the collision probability between random messages by two independent executions of $\Pi$.
We consider a collision between the quantum messages 
$\ket{\psi_1(x_1;r)}\otimes\ket{\psi_2(x_2;r)}$ in the first execution and the quantum messages 
$\ket{\psi_1(x^\prime_1;r^\prime)}\otimes\ket{\psi_2(x^\prime_2;r^\prime)}$ in the second execution, where 
$(x_1,x_2)$ and $(x^\prime_1,x^\prime_2)$ are 
independent random inputs under a distribution $\mu$,
and $r, r^\prime$ are independent shared random strings chosen uniformly at random.

However, the collision probability $\Pr[\ket{\psi(x_1,x_2;r)}=\ket{\psi(x^\prime_1,x^\prime_2;r^\prime}]$ of the quantum messages is not lower-bounded by the inverse of dimension of the message space as in the classical setting in \cite{AHMS20JCRYPT}.
Instead of the collision between two independent messages, we exploit the notion of the {\em purity}\/ of a quantum state $\rho$,
which is defined as $\tr\rho^2$. For this measure, it is known that the purity is lower-bounded by the inverse 
of the dimension of $\rho$.
}
As mentioned above, we use the purity as a collision measure of quantum messages in order to 
obtain lower bounds of quantum message length from upper bounds of the purity.
Indeed, we have for any quantum state $\rho$ with eigenvalues $\lambda_i>0$ $(i\in[\rank(\rho)])$,
\begin{equation}\label{eq:purity_lb}
 \tr\rho^2 = \sum_{i=1}^{\rank(\rho)}\lambda_i^2 \ge \frac{1}{\rank(\rho)}\left(\sum_{i=1}^{\rank(\rho)}\lambda_i\right)^2 = \frac{1}{\rank(\rho)}
\end{equation}
by the Cauchy--Schwarz inequality.

We set $\rho:=\sum_{x_1,x_2}\mu(x_1,x_2)\rho(x_1,x_2)$. By Eq.~(\ref{eq:purity_lb}) 
\begin{align}
 \frac{1}{\dim(\MM)} 
 &\le \tr\rho^2 
 = \tr\left(\sum_{x_1,x_2}\mu(x_1,x_2)\rho(x_1,x_2)\right)^2 \nonumber\\
 &=
 \tr\sum_{x_1,x_2,x^\prime_1,x^\prime_2}\mu(x_1,x_2)\rho(x_1,x_2)\mu(x^\prime_1,x^\prime_2)\rho(x^\prime_1,x^\prime_2).\label{eq:1/dim(M)}
\end{align}

Then, the purity of $\rho$ is upper-bounded as follows.
\begin{claim}\label{clm:purity upper bound}
\begin{align*}
 \tr\rho^2 \le \beta(\mu)^{-1}\tr\sum_{(x_1,x_2)\ne(x^\prime_1,x^\prime_2)}\mu(x_1,x_2)\rho(x_1,x_2)\mu(x^\prime_1,x^\prime_2)\rho(x^\prime_1,x^\prime_2).
\end{align*}
\end{claim}
The proof of this claim is done by a combinatorial analysis 
of the trace as a quantum counterpart of the analysis 
of the collision probability in \cite{AHMS20JCRYPT}.
The detailed proof will be given in Appendix~\ref{sec:appendix}.

Next, we consider an upper bound of 
\begin{equation}\label{eq:left-hand_of_Claim3.1}
    \tr\sum_{(x_1,x_2)\ne(x^\prime_1,x^\prime_2)}\mu(x_1,x_2)\rho(x_1,x_2)\mu(x^\prime_1,x^\prime_2)\rho(x^\prime_1,x^\prime_2).
 %   =
 %   \sum_{r,r^\prime}\pi(r)\pi(r^\prime) \tr\sum_{(x_1,x_2)\ne(x^\prime_1,x^\prime_2)}\mu(x_1,x_2)\rho(x_1,x_2;r)\mu(x^\prime_1,x^\prime_2)\rho(x^\prime_1,x^\prime_2;r^\prime)
\end{equation}
Actually, we show that for every $r$ and $r^\prime$, 
\[
\tr\sum_{(x_1,x_2)\ne(x^\prime_1,x^\prime_2)}\mu(x_1,x_2)\rho(x_1,x_2;r)\mu(x^\prime_1,x^\prime_2)\rho(x^\prime_1,x^\prime_2;r^\prime)
\]
is at most $2\alpha(\mu)2^{-H_\infty(\mu)}$ 
as this implies that Eq.~(\ref{eq:left-hand_of_Claim3.1}) is also at most $2\alpha(\mu)2^{-H_\infty(\mu)}$ (by Eq.~(\ref{eq:rho(x_1,x_2)})). 
This completes the proof of Lemma~\ref{lem:LBmain} by Claim~\ref{clm:purity upper bound} and Eq.~(\ref{eq:1/dim(M)}).

Now we fix any $r$ and $r'$. 
From Eq.~(\ref{eq:rho(x_1,x_2;r)}) and the union bound, 
%for every $r$ and $r^\prime$,
we have 
\begin{align}
\lefteqn{ \tr\sum_{(x_1,x_2)\ne(x^\prime_1,x^\prime_2)}\mu(x_1,x_2)\rho(x_1,x_2;r)\mu(x^\prime_1,x^\prime_2)\rho(x^\prime_1,x^\prime_2;r^\prime)}\nonumber\\
 & = \sum_{(x_1,x_2)\ne(x^\prime_1,x^\prime_2)}\mu(x_1,x_2)\mu(x^\prime_1,x^\prime_2)\lvert \braket{\psi_1(x_1;r)}{\psi_1(x^\prime_1;r^\prime)}\rvert^2 \cdot 
 \lvert \braket{\psi_2(x_2;r)}{\psi_2(x^\prime_2;r^\prime)}\rvert^2.\nonumber\\
 %\lvert \braket{\psi(x_1,x_2;r)}{\psi(x^\prime_1,x^\prime_2;r^\prime)}\rvert^2\nonumber\\
 & \le
 \sum_{x_1\ne x^\prime_1,x_2,x^\prime_2}\mu(x_1,x_2)\mu(x^\prime_1,x^\prime_2)\lvert \braket{\psi_1(x_1;r)}{\psi_1(x^\prime_1;r^\prime)}\rvert^2 \cdot 
 \lvert \braket{\psi_2(x_2;r)}{\psi_2(x^\prime_2;r^\prime)}\rvert^2.\nonumber\\
 %\lvert \braket{\psi(x_1,x_2;r)}{\psi(x^\prime_1,x^\prime_2;r^\prime)}\rvert^2\nonumber\\
 &\ \ \ \ + \sum_{x_1,x^\prime_1,x_2\ne x^\prime_2}\mu(x_1,x_2)\mu(x^\prime_1,x^\prime_2)
 %\lvert \braket{\psi(x_1,x_2;r)}{\psi(x^\prime_1,x^\prime_2;r^\prime)}\rvert^2.
 \lvert \braket{\psi_1(x_1;r)}{\psi_1(x^\prime_1;r^\prime)}\rvert^2 \cdot 
 \lvert \braket{\psi_2(x_2;r)}{\psi_2(x^\prime_2;r^\prime)}\rvert^2.\label{eq:first-term-bounded}
\end{align}
It suffices to show that the first term of the right-hand of  Eq.~(\ref{eq:first-term-bounded}) is at most $\alpha(\mu)2^{-H_\infty(\mu)}$ from the symmetry of $P_1$ and $P_2$.
%Recalling $|\psi(x_1,x_2;r)\rangle=|\psi_1(x_1;r)\rangle\otimes|\psi_2(x_2;r)\rangle$, the first term is equal to 
%\begin{align*}
%\[
%% & %\sum_{r,r^\prime} \pi(r)\pi(r^\prime)
%% \sum_{x_1\ne x^\prime_1,x_2,x^\prime_2}\mu(x_1,x_2)\mu(x^\prime_1,x^\prime_2)\lvert \braket{\psi(x_1,x_2;r)}{\psi(x^\prime_1,x^\prime_2;r^\prime)}\rvert^2\\
%% & = %\sum_{r,r^\prime}\pi(r)\pi(r^\prime)
% \sum_{x_1\ne x^\prime_1,x_2,x^\prime_2}\mu(x_1,x_2)\mu(x^\prime_1,x^\prime_2) \lvert \braket{\psi_1(x_1;r)}{\psi_1(x^\prime_1;r^\prime)}\rvert^2 \cdot  \lvert \braket{\psi_2(x_2;r)}{\psi_2(x^\prime_2;r^\prime)}\rvert^2.\]
%\end{align*}

In the classical case of \cite{AHMS20JCRYPT}, Applebaum et al.~used the fact that 
for every $x_1$ and every $r, r^\prime$ there exists at most one $z$ such that 
$\ket{\psi_1(x_1;r)}=\ket{\psi_1(z;r^\prime)}$ 
if $\ket{\psi_1(x_1;r)}, \ket{\psi_1(z;r^\prime)}$ are classical; that is, either
$\braket{\psi_1(x_1;r)}{\psi_1(z;r^\prime)}=0$ or $\braket{\psi_1(x_1;r)}{\psi_1(z;r^\prime)}=1$,
which can be derived from the non-degeneracy of $F_n$. 
However, we cannot demonstrate the same fact for quantum messages.
Instead, we can prove the following relaxed version of the fact for quantum messages.
\begin{lemma}\label{lem:weight}
If $F_n$ is non-degenerate,
we have
\[
 \sum_{z\in\XX_i} \lvert\braket{\psi_i(x_i;r)}{\psi_i(z;r^\prime)}\rvert^2 \le 1,
 \qquad
 \sum_{z\in\XX_i} \lvert\braket{\psi_i(z;r)}{\psi_i(x_i;r^\prime)}\rvert^2 \le 1
\]
for every $i\in\{1,2\}$, every $r,r^\prime$, and every $x_i\in\XX_i$.
Thus, the row sums and column sums of the weight matrix are at most $1$.
\end{lemma}
The proof of this lemma will be given in Appendix~\ref{sec:appendix}.

Fix $r$ and $r^\prime$ arbitrarily.
Let $w_1(x_1,x^\prime_1):=\lvert\braket{\psi_1(x_1;r)}{\psi_1(x^\prime_1;r^\prime)}\rvert^2$
and let $w_2(x_2,x^\prime_2):=\lvert\braket{\psi_2(x_2;r)}{\psi_2(x^\prime_2;r^\prime)}\rvert^2$.
We say that $x_1$ collides with $x^\prime_1$ if 
$w_1(x_1,x^\prime_1)>0$. Similarly, we say that $x_2$ collides with $x^\prime_2$ if 
$w_2(x_2,x^\prime_2)>0$.

Now, our final goal is to upper-bound 
\begin{equation}\label{eq:goal}
 %\sum_{r,r^\prime}\pi(r)\pi(r^\prime)
 \sum_{x_1\ne x^\prime_1,x_2,x^\prime_2}\mu(x_1,x_2)\mu(x^\prime_1,x^\prime_2) 
 w_1(x_1,x^\prime_1)w_2(x_2,x^\prime_2).    
\end{equation}

Let $W_1$ be the matrix $(w_1(x_1,x_1^\prime))_{x_1,x_1^\prime}$ with its diagonal
entries replaced by zero, and let $W_2=(w_2(x_2,x_2^\prime))_{x_2,x_2^\prime}$.
By Lemma~\ref{lem:weight}, both matrices are doubly substochastic.
Mirsky's theorem states that every square doubly substochastic matrix
is a convex combination of partial permutation matrices, that is,
zero--one matrices with at most one nonzero entry in each row and each
column \cite[Theorem~2.C.2, p.~38]{MOA11}. Hence, we can write
\[
 W_1=\sum_s\lambda_s A_s,\qquad W_2=\sum_t\eta_t B_t,
 \qquad \lambda_s,\eta_t>0,\qquad
 \sum_s\lambda_s=\sum_t\eta_t=1,
\]
where $A_s$ and $B_t$ are partial permutation matrices. 
Note that the zero matrix is also allowed as a partial permutation matrix.
Since all entries and coefficients are
nonnegative, every nonzero entry of $A_s$ or $B_t$ is a positive entry
of $W_1$ or $W_2$, respectively. In particular, every $A_s$ has zero
diagonal.

Fix any pair $A_s,B_t$. If either matrix is zero, its contribution is
zero. Otherwise, enumerate their nonzero entries as
$(u_i,u_i^\prime)$ for $i\in[k]$ and $(v_j,v_j^\prime)$ for $j\in[\ell]$,
respectively, and set
\[
 \vect{x_1}=(u_1,\ldots,u_k),\quad
 \vect{x_1}^\prime=(u_1^\prime,\ldots,u_k^\prime),\quad
 \vect{x_2}=(v_1,\ldots,v_\ell),\quad
 \vect{x_2}^\prime=(v_1^\prime,\ldots,v_\ell^\prime).
\]
Each list has distinct elements by a property of the partial permutation.
Moreover, $u_i$ collides with $u_i^\prime$ and $u_i\ne u_i^\prime$ for
every $i$, and $v_j$ collides with $v_j^\prime$ for every $j$.

Then, we can show that for every pair $A_s,B_t$, we have
\[
 \sum_{i,j} \mu(u_i,v_j)\mu(u^\prime_i,v^\prime_j) \le \alpha(\mu)2^{-H_\infty(\mu)}.
\]
The reason is as follows. We consider two rectangles  
$\RR:=(\vect{x_1},\vect{x_2})$ and $\RR^\prime:=(\vect{x^\prime_1},\vect{x^\prime_2})$.
Assume $R(\ket{\psi_1(u_i;r)}\ket{\psi_2(v_j;r)})\ne R(\ket{\psi_1(u^\prime_i;r^\prime)}\ket{\psi_2(v^\prime_j;r^\prime)})$, 
$u_i$ collides with $u^\prime_i$ and $v_j$ collides with $v^\prime_j$, where $R(\ket{\varphi})$ denotes the classical value that 
$R$ outputs on input $\ket{\varphi}$.
Then, $R$ can distinguish the two input states with certainty from the perfect correctness. This means
$\ket{\psi_1(u_i;r)}\ket{\psi_2(v_j;r)}$ and $\ket{\psi_1(u^\prime_i;r^\prime)}\ket{\psi_2(v^\prime_j;r^\prime)}$ must be orthogonal.
However, $w_1(u_i,u_i^\prime)w_2(v_j,v_j^\prime)>0$ since $u_i$ collides with $u^\prime_i$ and $v_j$ collides with $v^\prime_j$.
This is a contradiction. Therefore, $R(\ket{\psi_1(u_i;r)}\ket{\psi_2(v_j;r)})= R(\ket{\psi_1(u^\prime_i;r^\prime)}\ket{\psi_2(v^\prime_j;r^\prime)})$
if $u_i$ collides with $u^\prime_i$ and $v_j$ collides with $v^\prime_j$.

Therefore, $\XX_1$-disjoint rectangles 
$\RR$ and $\RR^\prime$
are similar; that is, $F_n[\RR]=F_n[\RR^\prime]$. 
Since $(\RR,\RR^\prime)$ is a pair of similar disjoint rectangles,
$\min\{\mu(\RR),\mu(\RR^\prime)\}\le\alpha(\mu)$ by the definition of $\alpha(\mu)$.
Using $\max_{(x_1,x_2)}\mu(x_1,x_2)=2^{-H_\infty(\mu)}$, we have
\begin{align*}
 \sum_{i,j}\mu(u_i,v_j)\mu(u^\prime_i,v^\prime_j)
 &\le 2^{-H_\infty(\mu)}\sum_{i,j}\mu(u_i,v_j)=2^{-H_\infty(\mu)}\mu(\RR),\\
 \sum_{i,j}\mu(u_i,v_j)\mu(u^\prime_i,v^\prime_j)
 &\le 2^{-H_\infty(\mu)}\sum_{i,j}\mu(u^\prime_i,v^\prime_j)=2^{-H_\infty(\mu)}\mu(\RR^\prime),
\end{align*}
where the second equalities use that the entries of each list are distinct.
Hence
\[
 \sum_{i,j}\mu(u_i,v_j)\mu(u^\prime_i,v^\prime_j)
 \le 2^{-H_\infty(\mu)}\min\{\mu(\RR),\mu(\RR^\prime)\}
 \le 2^{-H_\infty(\mu)}\alpha(\mu).
\]

Averaging this bound over the two convex combinations, we have
\begin{align*}
 & \sum_{x_1\ne x^\prime_1,x_2,x^\prime_2}\mu(x_1,x_2)\mu(x^\prime_1,x^\prime_2)
 w_1(x_1,x^\prime_1)w_2(x_2,x^\prime_2)\\
 & = \sum_{s,t}\lambda_s\eta_t
 \sum_{\substack{u,u^\prime,v,v^\prime\\A_s(u,u^\prime)=B_t(v,v^\prime)=1}}
 \mu(u,v)\mu(u^\prime,v^\prime)
  \le \sum_{s,t}\lambda_s\eta_t\,\alpha(\mu)2^{-H_\infty(\mu)}
 =\alpha(\mu)2^{-H_\infty(\mu)}.
\end{align*}
Eventually, an upper bound of Eq.~(\ref{eq:goal}) is  $\alpha(\mu)2^{-H_\infty(\mu)}$.
\end{proof}

\section{Power of Shared Entanglement for Total Functions}

In this section, we prove Theorem~\ref{thm:upperbound}, 
which implies a factor two gap between PSQMs with shared entanglement and without shared entanglement for a total function. The main part of Theorem~\ref{thm:upperbound} provides a $k$-party PSQM protocol for a total function $GEQ_{2l}: (\{0,1\}^{2l})^k\rightarrow\{0,1\}$ 
defined by 
\[
GEQ_{2l}(x_1,x_2,\ldots,x_k)=
\left\{
\begin{array}{ll}
1 & (\sum_{j=1}^k x_j^1=\sum_{j=1}^k x_j^2
=\cdots=\sum_{j=1}^k x_j^{2l-1}=\sum_{j=1}^k x_j^{2l}=0),\\
0 & (\mbox{otherwise}),
\end{array}
\right.
\] 
where each $x_j=x_j^1x_j^2\cdots x_j^{2l-1}x_j^{2l}$ is an element of $\{0,1\}^{2l}$, and the summation is taken over $\mathbb{F}_2$.
To reduce the communication complexity from the trivial $2kl$ qubits 
to $kl$ qubits, we encode half of the input bits 
by bit flipping of the shared state $\frac{1}{\sqrt{2}}(|0^k\rangle+|1^k\rangle)$ (called the $k$-qubit GHZ state) among the $k$ parties,
and the other half by phase flipping of the state, by a method similar to superdense coding (e.g., see~\cite{NC00}). More precisely, 
we exploit an encoding similar to two-party quantum SMP protocols with shared entangled states to compute the equality function~\cite{HSW+05}. 
However, this is not sufficient for PSQM protocols.
To convert quantum SMP protocols into PSQM protocols,
we further use shared randomness among the $k$ parties,
and hide the input strings from the referee except for the output of the function $GEQ_{2l}$. This hiding can be shown to be possible by multiplying a random element in $\mathbb{F}_{2^{2l}}^*$ by the element in $\mathbb{F}_{2^{2l}}$ that corresponds to the input $x_j$. 

For the proof of Theorem~\ref{thm:upperbound},
we first consider a PSQM protocol for a finite function. 
Let $Sum_2:(\{0,1\}^2)^k\rightarrow\{0,1\}^2$ be 
\[
Sum_2(x_1,x_2,\ldots,x_k)=\left(\sum_{j=1}^k x_j^1, \sum_{j=1}^k x_j^2\right),
\] 
where each $x_j=x_j^1x_j^2$ is an element of $\{0,1\}^2$, 
and the summation is taken over $\mathbb{F}_2$.

\begin{lemma}\label{lemma:M2}
For any even (resp. odd) $k$, $Q^{psm,*}_{0}(Sum_2)\leq k$ (resp. $\leq k+1$). 
\end{lemma}

\begin{proof}
First, we consider the case in which $k$ is even. The quantum protocol is as follows.

\ 

\noindent
{\bf Protocol ${\cal P}_{Sum_2}$:}

0. All the parties share the entangled state
\[
\frac{1}{\sqrt{2}}\left(\bigotimes_{j=1}^k |0\rangle_{Q_j} + \bigotimes_{j=1}^k |1\rangle_{Q_j}\right), 
\]
where the single-qubit register $Q_j$ is owned by party $P_j$. Moreover, the parties share a $k$-bit string $r=r_1r_2\cdots r_k$ such that $\sum_{j=1}^k r_j=0$ uniformly at random.

1. Each party $P_j$ applies $Z$ on $Q_j$ if $x_j^2=1$. The resulting state is
\[
\frac{1}{\sqrt{2}}\left(\bigotimes_{j=1}^k |0\rangle_{Q_j} + \bigotimes_{j=1}^k (-1)^{x_j^2} |1\rangle_{Q_j}\right)=\frac{1}{\sqrt{2}}\left(\bigotimes_{j=1}^k |0\rangle_{Q_j} +  (-1)^{\sum_{j=1}^kx_j^2} \bigotimes_{j=1}^k  |1\rangle_{Q_j}\right).
\]

2. Each party $P_j$ applies $X$ on $Q_j$ if $x_j^1\oplus r_j=1$. The resulting state is
\begin{equation}\label{eq:M2_step2}
\frac{1}{\sqrt{2}}\left(\bigotimes_{j=1}^k |x_j^1\oplus r_j \rangle_{Q_j} + (-1)^{\sum_{j=1}^k x_j^2} \bigotimes_{j=1}^k |x_j^1\oplus r_j\oplus 1\rangle_{Q_j}\right).
\end{equation}

3. Each party $P_j$ sends $Q_j$ to the referee $R$.

4. $R$ measures quantum registers $Q_1,Q_2,\ldots,Q_k$ in the basis 
\[
\left\{
|\Phi(y_1,y_2,\ldots,y_{k-1},z)\rangle:=\frac{1}{\sqrt{2}}\left(|y_1,y_2, \ldots, y_{k-1}, 0\rangle + (-1)^{z} |y_1\oplus 1, y_2\oplus 1, \ldots, y_{k-1}\oplus 1, 1\rangle \right)
\right\},
\] 
and let $y_1y_2\cdots y_{k-1}z$ be the measurement result.

5. $R$ outputs the two bits $\sum_{j=1}^{k-1} y_j$ and $z$.

\

{\bf Correctness:} 
The second bit of the output of $R$ is $z=\sum_{j=1}^k x_j^2$, as desired.
For the first bit, we consider two cases: (i) $x_k^1\oplus r_k=0$ and (ii) $x_k^1\oplus r_k=1$. 
We first consider case (i). Then, $y_j=x_j^1\oplus r_j$ for $j=1,2,\ldots,k-1$, and we thus obtain the desired output
\[
\sum_{j=1}^{k-1} y_j = \sum_{j=1}^{k-1} (x_j^1\oplus r_j) = \sum_{j=1}^{k} (x_j^1\oplus r_j) = \sum_{j=1}^{k} x_j^1,
\]
where the second equality originates from $x_k^1\oplus r_k=0$, and the last equality originates from $\sum_{j=1}^{k} r_j=0$. 
For case (ii), $y_j=x_j^1\oplus r_j\oplus 1$ for $j=1,2,\ldots,k-1$, and we thus obtain the desired output
\[
\sum_{j=1}^{k-1} y_j 
= \sum_{j=1}^{k-1} (x_j^1\oplus r_j\oplus 1) = \left(\sum_{j=1}^{k-1} x_j^1\oplus r_j\right) \oplus 1
= \sum_{j=1}^{k} (x_j^1\oplus r_j) = \sum_{j=1}^{k} x_j^1,
\]
where the second equality originates from the fact that $k$ is even, the third equality originates from $x_k^1\oplus r_k=1$, and the last equality originates from $\sum_{j=1}^{k} r_j=0$.

{\bf Privacy:} Let the output of $Sum_2$ be $(b_1,b_2)$,  
where $b_1=\sum_{j=1}^k  x_j^1$ and $b_2=\sum_{j=1}^k x_j^2$. 
As $R$ has no knowledge about $r_1,\ldots,r_k$ except that the sum is $0$, we can observe that the quantum state that $R$ receives (represented by Eq.~(\ref{eq:M2_step2})) is taken from the set of $2^{k-2}$ orthogonal pure states 
\[
\left\{|\Phi(y_1,\ldots,y_{k-1},b_2)\rangle:~\sum_{j=1}^{k-1} y_j=b_1\right\}
\] 
(up to the total phase) uniformly at random. Thus, the simulator can simulate the distribution of the message given the output of $Sum_2$. 

In the case in which $k$ is odd, the last party $P_k$ prepares an extra two-bit string $x_{k+1}=00$, and ${\cal P}_{Sum_2}$ is run for the $(k+1)$-party case,
where $P_k$ also plays the role of the party $P_{k+1}$. 
\end{proof}

Next, we present the main lemma for Theorem~\ref{thm:upperbound}. 

\begin{lemma}\label{thm:GEQ_2l}
For any even (resp. odd) $k$, 
$Q^{psm,*}_{0}(GEQ_{2l})\leq kl$ (resp.~$\leq (k+1)l$). 
\end{lemma}

\begin{proof}
We only demonstrate the case in which $k$ is even 
(as the odd case is considered similarly to the proof of Lemma~\ref{lemma:M2}). 
The quantum protocol is as follows.

\

\noindent
{\bf Protocol ${\cal P}_{GEQ_{2l}}$:}

0. All parties share the entangled state
\[
\bigotimes_{i=1}^l 
\left(
\frac{1}{\sqrt{2}}\left(\bigotimes_{j=1}^k |0\rangle_{Q_j^i} + \bigotimes_{j=1}^k |1\rangle_{Q_j^i}\right)
\right), 
\]
where the single-qubit registers $Q_j^1,\ldots,Q_j^l$ are owned by party $P_j$. 
Moreover, they share $l$ $k$-bit strings $r^i:=r_1^ir_2^i\cdots r_k^i$ such that $\sum_{j=1}^k r_j^i=0$ ($i=1,2,\ldots,l$) uniformly at random, and a non-zero $2l$-bit string $r'=r'_1r'_2\cdots r'_{2l-1}r'_{2l}$ uniformly at random.

1. Each party $P_j$ computes the $2l$-bit string $a_j=a_j^1a_j^2\cdots a_j^{2l}$ defined as ${\sf p}(a_j)={\sf p}(r'){\sf p}(x_j)$.

2. Each party $P_j$ applies $Z$ on $Q_j^i$ if $a_j^{2i}=1$. 
The resulting state is
\[
\bigotimes_{i=1}^l
\left(
\frac{1}{\sqrt{2}}\left(\bigotimes_{j=1}^k |0\rangle_{Q_j^i} + (-1)^{\sum_{j=1}^k a_j^{2i}} \bigotimes_{j=1}^k  |1\rangle_{Q_j^i}\right)
\right).
\]

3. Each party $P_j$ applies $X$ on $Q_j^i$ if $a_j^{2i-1}\oplus r_j^i=1$. The resulting state is
\begin{equation}\label{eq:GEQ2l_step3}
\bigotimes_{i=1}^l
\left(
\frac{1}{\sqrt{2}}\left(\bigotimes_{j=1}^k |a_j^{2i-1}\oplus r_j^i \rangle_{Q_j^i} + (-1)^{\sum_{j=1}^k a_j^{2i}} \bigotimes_{j=1}^k |a_j^{2i-1}\oplus r_j^i \oplus 1\rangle_{Q_j^i}\right)
\right).
\end{equation}

4. Each party $P_j$ sends $l$ quantum registers $Q_j^1,\ldots,Q_j^l$ to $R$.

5. $R$ measures $kl$ quantum registers $Q_1^1,\ldots,Q_k^1,\ldots,Q_1^l,\ldots,Q_k^l$ in the basis 
\[
B:=\left\{
\bigotimes_{i=1}^l |\Phi(y_1^i,\ldots,y_{k-1}^i,z^i)\rangle
:~y_1^i\cdots y_{k-1}^iz^i\in\{0,1\}^k~\mbox{for every}~i\in [l]
\right\},
\]
and let $y_1^iy_2^i\cdots y_{k-1}^iz^i$ ($i=1,\cdots,l$) 
be the measurement results.

6. $R$ accepts if $\left(\sum_{j=1}^{k-1} y_j^i\right)=z^i=0$ 
for all $i=1,\cdots,l$ and rejects otherwise.

\

{\bf Correctness:} Note that $(\sum_{j=1}^k x_j^1, \ldots, \sum_{j=1}^k x_j^{2l})=(0,\ldots,0)$ if and only if \sloppy
$(\sum_{j=1}^k a_j^1, \ldots, \sum_{j=1}^k a_j^{2l})=(0,\ldots,0)$, since 
\begin{equation}\label{eq:GEQ2l_correctness}
{\sf p}\left((\sum_{j=1}^k a_j^1, \ldots, \sum_{j=1}^k a_j^{2l}) \right)
=
%\sum_{j=1}^k p( (a_{j1},a_{j2}) )=\sum_{j=1}^k p(r')p( (x_{j1},x_{j2}) )=
{\sf p}(r'){\sf p}\left( (\sum_{j=1}^k x_j^1, \ldots, \sum_{j=1}^k x_j^{2l}) \right).
\end{equation}
Now, the correctness of ${\cal P}_{Sum_2}$ in Lemma~\ref{lemma:M2} also guarantees the correctness of ${\cal P}_{GEQ_{2l}}$. 

{\bf Privacy:} First, we consider the case in which $GEQ_{2l}(x_1,x_2,\ldots,x_k)=1$: that is, $(\sum_{j=1}^k x_j^1, \ldots, \sum_{j=1}^k x_j^{2l})=(0,\ldots,0)$. Then, as demonstrated in Eq.~(\ref{eq:GEQ2l_correctness}), $(a_1,a_2,\ldots, a_k)$ satisfies $(\sum_{j=1}^k a_j^1, \ldots, \sum_{j=1}^k a_j^{2l})=(0,\ldots,0)$.
Moreover, $(a_1^{2i-1},a_2^{2i-1},\ldots,a_k^{2i-1})$ is uniformly randomized by $r^i$ in Step 3 under the restriction that the sum is $0$.  
Thus, the quantum state that $R$ receives (represented by Eq.~(\ref{eq:GEQ2l_step3})) is taken from the set of $2^{(k-2)l}$ orthogonal pure states 
\[
\left\{
\bigotimes_{i=1}^l |\Phi(y_1^i,\ldots,y_{k-1}^i,0)\rangle:~\sum_{j=1}^{k-1} y_j^i=0~\mbox{for every}~i\in [l] 
\right\}
\] 
(up to the total phase) uniformly at random. 
Second, we consider the case in which $GEQ_{2l}(x_1,x_2,\ldots,x_k)=0$, 
i.e., $(\sum_{j=1}^k x_j^1, \ldots, \sum_{j=1}^k x_j^{2l})$ is in 
the set 
\[
S:=\{(b_1,\ldots,b_{2l}):~b_1\cdots b_{2l}\in\{0,1\}^{2l}\} \setminus\{(0,\ldots,0)\}.
\]
Then, by Eq.~(\ref{eq:GEQ2l_correctness}), 
$(a_1,a_2,\ldots, a_k)$ is taken so that  
$(\sum_{j=1}^k a_j^1, \ldots, \sum_{j=1}^k a_j^{2l})$ can be distributed from $S$ uniformly at random. 
Moreover, $(a_1^{2i-1},a_2^{2i-1},\ldots,a_k^{2i-1})$ is uniformly randomized by $r^i$ in Step 3 under the restriction that the sum remains the same (since $\sum_{j=1}^k r_j^i=0$). 
Thus, the quantum state that $R$ receives (represented by Eq.~(\ref{eq:GEQ2l_step3})) is taken from the set of 
$(2^{2l}-1) 2^{(k-2)l}$ 
orthogonal pure states 
\[
B\setminus
\left\{
\bigotimes_{i=1}^l |\Phi(y_1^i,\ldots,y_{k-1}^i,0)\rangle:~\sum_{j=1}^{k-1} y_j^i=0~\mbox{for every}~i\in [l] 
\right\}
%\left\{|\Phi(y_1,\ldots,y_{k-1},0)\rangle:~\sum_{j=1}^{k-1} y_j=1\right\} \cup \left\{|\Phi(y_1,\ldots,y_{k-1},1)\rangle:~\sum_{j=1}^{k-1} y_j=0,1 \right\}
\] 
%\[\left\{
%\frac{1}{\sqrt{2}} \left( |y_1,y_2,\ldots,y_{k-1},0\rangle + (-1)^{z} |y_1\oplus 1,y_2\oplus 1,\ldots,y_{k-1}\oplus 1,1\rangle \right):~\sum_{j=1}^{k-1} y_j=0,\ z\in\{0,1\}
%\right\}\]
(up to the total phase) uniformly at random. 
\end{proof}

Now Lemma~\ref{thm:GEQ_2l} provides the upper bound $kn/2$ of Theorem~\ref{thm:upperbound}. 
The lower bound $kn$ of Theorem~\ref{thm:upperbound} originates from the lower bound $n$ of the exact (two-party) one-way quantum communication complexity with no shared entanglement of the $n$-bit equality function (see, e.g.,~\cite[Theorem 5.11]{Kla07}) as it implies that for any $j\in [k]$, the $j$th party must send $n$ qubits 
(considering the one-way communication setting from the $j$th party with input $x\in\{0,1\}^n$ to the group of the referee and the other $k-1$ parties in which one party has input $y\in\{0,1\}^n$ and the $k-2$ remaining parties have input $0^n$, the length of the message of the $j$th party must be $n$).  
For completeness, the matching upper bound without shared
entanglement is achieved by a classical PSM protocol. Identify each input with an element of $\mathbb F_{2^n}$. The parties share an independent uniform element $a\in\mathbb F_{2^n}^*$ and a uniform tuple $(r_1,\ldots,r_k)$ satisfying $\sum_j r_j=0$.
Party $P_j$ sends $m_j=ax_j+r_j$, and the referee accepts if
$\sum_jm_j=0$. Correctness follows from $a\ne0$.
The joint messages are uniform over the tuples with zero sum on
yes-inputs and uniform over the tuples with nonzero sum on
no-inputs. Thus the protocol has perfect privacy and communicates
exactly $kn$ bits.
This completes the proof of Theorem~\ref{thm:upperbound}.

Actually, the upper bound $kl$ of Lemma~\ref{thm:GEQ_2l} for $GEQ_{2l}$ 
 is tight when $k$ is even. 
The matching lower bound $kl$ originates from the lower bound $l$ of the exact one-way quantum communication complexity with shared entanglement of the $2l$-bit equality function 
shown by Klauck~\cite[Theorem 5.12]{Kla07} as it   implies that each party must send $l$ qubits.

\section{Power of Shared Entanglement for Partial Functions}

In this section, we prove Theorem~\ref{thm:exp-gap-DJ} 
(for a partial function, correctness and privacy are required only for inputs in the domain). We consider the so-called distributed Deutsch-Jozsa problem $DJ_n$ introduced by Brassard, Cleve, and Tapp~\cite{BCT99PRL}. 
First we show that $C_0^{psm,*}(DJ_n)=O(\log n)$. 
Our PSM protocol is based on the protocol provided in~\cite{BCT99PRL}, 
which we modify so that the privacy condition can be satisfied. 
Second we show $Q_0^{psm}(DJ_n)=\Omega(n)$ by showing that every exact quantum SMP protocol for $DJ_n$ requires $\Omega(n)$ communication, even without the privacy requirement.

%by observing that the fact that the exact classical and quantum SMP communication complexities are the same for total functions can be extended to the case of partial functions. 

Let $n$ be any power of $2$. The distributed Deutsch-Jozsa problem $DJ_n:\{0,1\}^n\times\{0,1\}^n\rightarrow\{0,1\}$, introduced in~\cite{BCT99PRL}, is defined as
\[
DJ_n(x,y)=
\left\{
\begin{array}{ll}
1 & \mbox{if}~~x=y \\
0 & \Delta(x,y)=n/2, 
\end{array} 
\right.
\]
where $\Delta(x,y)$ denotes the Hamming distance between $x=x_0x_1\cdots x_{n-1}$ and $y=y_0y_1\cdots y_{n-1}$.

\begin{lemma}\label{thm:DJ}
There is a PSM protocol with shared entanglement that solves $DJ_n$ with probability $1$, and the classical communication complexity is $2\log n$. 
\end{lemma}

\begin{proof}
The PSM protocol is as follows.

\

\noindent
{\bf Protocol ${\cal P}_{DJ}$:} Let $n=2^m$ (with $m\geq 1$).

0. $P_1$ and $P_2$ share the entangled state
\[
\frac{1}{\sqrt{n}} \sum_{i\in\{0,1\}^{m}} |i\rangle_A|i\rangle_B
\]
and the two $m$-bit strings $r\in\{0,1\}^m\setminus\{0^m\}$ and $r'\in\{0,1\}^m$ which are taken uniformly at random.

1. $P_1$ (resp.~$P_2$) adds phase $(-1)^{x_i}$ ($(-1)^{y_i}$) to the $m$-qubit register $A$ ($B$) if the content of $A$ ($B$) is $i$ (where $i\in\{0,1\}^m$ is identified as the corresponding non-negative integer). 
The resulting state is
\[
\frac{1}{\sqrt{n}} \sum_{i\in\{0,1\}^{m}} (-1)^{x_i}|i\rangle_A (-1)^{y_i}|i\rangle_B.
\]

2. $P_1$ and $P_2$ apply the Hadamard gate $H$ for each qubit 
of their registers $A$ and $B$, respectively. The resulting state is
\begin{equation}\label{eq:DJ_step2}
\frac{1}{n\sqrt{n}} \sum_{i\in\{0,1\}^{m}} 
\left(
(-1)^{x_i}\sum_{k\in\{0,1\}^{m} } (-1)^{i\cdot k} |k\rangle_A
\right)
\left(
(-1)^{y_i}\sum_{l\in\{0,1\}^{m} } (-1)^{i\cdot l} |l\rangle_B
\right).
\end{equation}

3. $P_1$ and $P_2$ measure $A$ and $B$ in the computational basis, respectively, 
and let $K$ and $L$ be the resulting bit strings in $\{0,1\}^{m}$.

4. $P_1$ and $P_2$ send classical messages $m_A$ and $m_B$ defined as ${\sf p}(m_A)={\sf p}(r){\sf p}(K)+{\sf p}(r')$ 
and ${\sf p}(m_B)={\sf p}(r){\sf p}(L)+{\sf p}(r')$ to $R$, respectively.

5. $R$ accepts if $m_A=m_B$ and rejects otherwise.

\

{\bf Correctness:} 
Note that the amplitude of $|k\rangle_A|l\rangle_B$ in Eq.~(\ref{eq:DJ_step2}) 
is
\begin{equation}\label{eq:DJ_step3}
\frac{1}{n^{3/2}} \sum_{i\in\{0,1\}^{m}} (-1)^{x_i\oplus y_i} (-1)^{i\cdot (k\oplus l)}.
\end{equation}
When $x=y$, Eq.~(\ref{eq:DJ_step3}) is $0$ if $K\neq L$. 
Thus, the event $K=L$ occurs with probability $1$; therefore, $R$ always accepts.
When $\Delta(x,y)=n/2$, Eq.~(\ref{eq:DJ_step3}) is $0$ if $K=L$. 
Thus, the event $K\neq L$ occurs with probability $1$; therefore, $R$ always rejects.

{\bf Privacy:}
Again, by Eq.~(\ref{eq:DJ_step3}),
if $x=y$, then $K=L=k$ is obtained with $1/n$ for each $k\in\{0,1\}^m$.
Thus, the simulator can simulate the messages 
by generating the same $m$-bit string chosen uniformly at random as $P_1$'s and $P_2$'s messages.   
If $\Delta(x,y)=n/2$, 
the element ${\sf p}(K)-{\sf p}(L)$ in $\mathbb{F}_{2^m}$ is nonzero; thus, the difference between ${\sf p}(r){\sf p}(K)+{\sf p}(r')$ and ${\sf p}(r){\sf p}(L)+{\sf p}(r')$ is distributed uniformly at random in $\mathbb{F}_{2^m}^*$. 
Moreover, ${\sf p}(K)$ (and ${\sf p}(L)$) is distributed uniformly at random in $\mathbb{F}_{2^m}$ by multiplying ${\sf p}(r)$ and adding ${\sf p}(r')$.
Thus, the simulator can simulate the messages by choosing two different $m$-bit strings uniformly at random as $P_1$'s and $P_2$'s messages. 
\end{proof}

Using a proof idea of the result in~\cite{BCW98STOC} that $DJ_n$ has the exact classical communication complexity $\Omega(n)$ (even in the two-way communication model), 
we can show that $DJ_n$ provides the following exponential separation between exact PSMs with shared entanglement and exact PSQMs without shared entanglement. (Note that Theorem~\ref{thm:PSM-DJ} implies Theorem~\ref{thm:exp-gap-DJ}, namely, an exponential gap between $Q^{psm,*}_0(DJ_n)$ and $Q^{psm}_0(DJ_n)$, as well as between $C^{psm,*}_0(DJ_n)$ and $C^{psm}_0(DJ_n)$.)

\begin{theorem}\label{thm:PSM-DJ}
$C^{psm,*}_0(DJ_n)=O(\log n)$ and $Q^{psm}_0(DJ_n)=\Omega(n)$.
\end{theorem}

\begin{proof}
The upper bound $C^{psm,*}_0(DJ_n)=O(\log n)$ follows from Lemma~\ref{thm:DJ}. 

For the lower bound $Q^{psm}_0(DJ_n)=\Omega(n)$, 
we show the communication complexity is $\Omega(n)$ 
even if the privacy condition is ignored. 
We use the forbidden-intersection bound of Frankl and R\"odl~\cite[Corollary~1.2]{FR87}, as in~\cite{BCW98STOC}. 
It suffices to prove the lower bound for sufficiently large $n$. Since $n$ is a power of two, such an $n$ is divisible by $4$.

Since we consider only perfect correctness, 
we can omit the shared randomness from the notation and 
can assume that their messages are pure. 
%since the following argument applies after fixing any value in the support of the shared random string. 
Let $P_1$'s and $P_2$'s messages be $|a_x\rangle\in\mathcal H_A$ and $|b_y\rangle\in\mathcal H_B$, where $\mathcal H_A$ and $\mathcal H_B$ are the corresponding Hilbert spaces, respectively. 
For each $x\in\{0,1\}^n$, define the unit vector $|c_x\rangle:=|a_x\rangle\otimes|b_x\rangle$, and let $C$ be its Gram matrix:
\begin{equation}
 C_{z,x}:=\langle c_z|c_x\rangle =\langle a_z|a_x\rangle\langle b_z|b_x\rangle.
 \label{eq:gram}
\end{equation}
Then $C$ is positive semidefinite, $C_{x,x}=1$ for every $x$, and $\rank(C) \le \dim(\mathcal H_A)\dim(\mathcal H_B)$. Let $V_{n/2}$ be the set of $n$-bit strings whose Hamming weight is $n/2$. 
%\[ V_{n/2}:=\{x\in\{0,1\}^n:\wt(x)=n/2\},
%% \qquad N:=|V_{n/2}|=\binom{n}{n/2},
% \label{eq:middle-layer}
%\]

Let $G$ be the principal submatrix of $C$ indexed by $V_{n/2}$. Since $G$ is itself a Gram matrix, it is positive semidefinite and  
\begin{equation}
% G\succeq0,
% \qquad
 G_{x,x}=1\quad(x\in V_{n/2}), \label{eq:G-basic}
\end{equation}
and
\begin{equation}
 \rank(G) \le \rank(C) \le \dim(\mathcal H_A)\dim(\mathcal H_B). \label{eq:rank-upper}
\end{equation}

For a fixed $z\in V_{n/2}$, define the nonzero support of the $z$th row by
\begin{equation}
 S_z:=\{x\in V_{n/2}:G_{z,x}\ne0\}.
 \label{eq:Sz}
\end{equation}
We claim that $S_z$ contains no two strings at Hamming distance $n/2$.

Suppose that $x,y\in S_z$.  
Since $G$ is a principal submatrix of $C$, Eqs.~\eqref{eq:gram} and \eqref{eq:Sz} imply
\[
 \langle a_z|a_x\rangle\ne0,
 \qquad
 \langle b_z|b_x\rangle\ne0,
\]
and
\[
 \langle a_z|a_y\rangle\ne0,
 \qquad
 \langle b_z|b_y\rangle\ne0.
\]
In particular,
\begin{equation}
 \langle a_z|a_x\rangle\langle b_z|b_y\rangle\ne0.
 \label{eq:cross-nonzero}
\end{equation}
If $\Delta(x,y)=n/2$, then $(z,z)$ is a yes-input of $DJ_n$ and $(x,y)$ is a no-input.  
Perfect correctness requires their joint message states to be perfectly distinguishable and hence orthogonal. 
Therefore,
\[ 
0 =\langle a_z\otimes b_z|a_x\otimes b_y\rangle = \langle a_z|a_x\rangle\langle b_z|b_y\rangle,
% \label{eq:exact-orthogonality}
\]
contradicting Eq.~\eqref{eq:cross-nonzero}. 
Hence
\[ 
x,y\in S_z,\ x\ne y \quad\Longrightarrow\quad \Delta(x,y)\ne n/2.
 %\label{eq:forbidden-distance}
\]

All strings in $S_z$ have Hamming weight $n/2$. 
Identifying each $x\in S_z$ with its support $\operatorname{supp}(x):=\{i: x_i=1\}$, 
we obtain a set family $\{\operatorname{supp}(x):x\in S_z\} \subseteq\binom{[n]}{n/2}$, 
where $\binom{[n]}{n/2}$ denotes the family of all $n/2$-element subsets of $[n]$. 
This family contains no distinct pair whose intersection has size $n/4$. 
Corollary 1.2 in \cite{FR87} (see also \cite[p.110]{Wol01}) therefore yields
\[
 |S_z|\le (1.99)^n.
 %\label{eq:row-support}
\]
Thus every row of $G$ has at most $(1.99)^n$ nonzero entries, and
\[
 \operatorname{nnz}(G) = \sum_{z\in V_{n/2}} |S_z| \le |V_{n/2}|(1.99)^n,
% \label{eq:nnz}
\]
where $\operatorname{nnz}(G)$ is the number of nonzero entries in $G$.

We now turn the row-sparsity bound into a rank bound.  
Since $G$ is positive semidefinite, we can use the following inequality for positive semidefinite matrices (by taking $\rho=G/\operatorname{tr}(G)$ in Eq.~(\ref{eq:purity_lb})): 
\begin{equation}
\rank(G) \ge \frac{\operatorname{tr}(G)^2}{\operatorname{tr}(G^2)}. \label{eq:trace-rank}
\end{equation}
%\[ \rank(G) \ge \frac{\operatorname{tr}(G)^2}{\operatorname{tr}(G^2)}  \] 
%if $\lambda_1,\ldots,\lambda_r>0$ are the nonzero eigenvalues of $G$, where $r=\rank(G)$, then Cauchy--Schwarz gives \[ \operatorname{tr}(G)^2 = \left(\sum_{i=1}^r\lambda_i\right)^2 \le r\sum_{i=1}^r\lambda_i^2 = \rank(G) \operatorname{tr}(G^2). \]
%Cauchy--Schwarz applied to its nonzero eigenvalues gives
By Eq.~\eqref{eq:G-basic},
\begin{equation}
 \operatorname{tr}(G)=|V_{n/2}|.
 \label{eq:trace}
\end{equation}
Every entry of a Gram matrix of unit vectors has absolute value at most one.
Therefore,
\begin{equation}
 \operatorname{tr}(G^2) = \sum_{z,x\in V_{n/2}}|G_{z,x}|^2 \le\operatorname{nnz}(G) \le |V_{n/2}|(1.99)^n. \label{eq:trace-square}
\end{equation}
Combining Eqs.~\eqref{eq:trace-rank}, \eqref{eq:trace}, and \eqref{eq:trace-square} with $|V_{n/2}|=\binom{n}{n/2}$, we obtain
\[
 \rank(G) \ge \frac{\binom{n}{n/2}}{(1.99)^n}.
\]
Using the standard central-binomial estimate $\binom{n}{n/2}=\Theta\!\left(\frac{2^n}{\sqrt n}\right)$, 
we conclude that
\[
 \rank(G)\ge \Omega\!\left(
  \frac{\left( \frac{2}{1.99} \right)^n}{\sqrt n}\right) = 2^{\Omega(n)}.
\]
By Eq.~\eqref{eq:rank-upper},
\[
 \dim(\mathcal H_A)\dim(\mathcal H_B) \ge2^{\Omega(n)}.
\]
This implies 
\[
 \log\dim(\mathcal H_A)
 +\log\dim(\mathcal H_B)
 =\Omega(n),
\]
which is the claimed total communication lower bound.
\end{proof}

Theorem~\ref{thm:PSM-DJ} provides an exponential gap between PSMs with shared entanglement and PSQMs without shared entanglement for partial functions; however, it is obtained only in the {exact} setting, and we do not know whether this exponential gap can be obtained in the bounded-error setting for partial or total functions. 
However, we can observe that there is a relational problem that has an exponential gap between PSMs with shared entanglement and PSQMs without shared entanglement in the bounded-error setting from the result by Gavinsky et al.~\cite{GKRW09SICOMP}. They demonstrated that the problem has an exponentially smaller communication complexity of a classical SMP protocol with shared entanglement than the communication complexity of quantum SMP protocols only with shared randomness, 
while it is easy to see that their protocol is in fact a PSQM.
\section{Conclusion}

This paper introduced a quantum analogue of the well-studied PSM model which was called the PSQM model, and provided several initial results in the exact setting.
Here we list a number of open problems.

\begin{itemize}
    \item Can the lower bound of Theorem~\ref{thm:lowerbound} be extended to the bounded-error setting or to the shared entanglement case? 
    \item Can we obtain any non-trivial communication complexity gap between the PSM model and the PSQM model for some function in the bounded-error setting? How about any non-trivial communication complexity gap between the PSQM model with shared entanglement and the PSQM model without it in the bounded-error setting?
    \item The PSQM model in this paper was defined only for perfect privacy. What results are obtained for imperfect privacy? To extend the lower bound of Theorem~\ref{thm:lowerbound} to imperfect privacy, a quantumly tailored modification of \cite[Section 5]{AHMS20JCRYPT} might be worth considering, while it seems much more complicated than the proof of Theorem~\ref{thm:lowerbound}. 
\end{itemize}

%%
%% Bibliography
%%

%% Please use bibtex, 

%\bibliography{lipics-v2021-sample-article}

\appendix

\section{Proofs of Claim~\ref{clm:purity upper bound} and Lemma~\ref{lem:weight}}\label{sec:appendix}

In this appendix, 
%we give the proofs of Claim~\ref{clm:referee is PVM}, Claim~\ref{clm:referee is PVM} and Lemma~\ref{clm:orthogonality}.
 we give the detailed proofs of the technical claim 
 and lemma used in the proof of Lemma~\ref{lem:LBmain}.
\ignore{
\begin{claim}\label{clm:purity upper bound}
\begin{align*}
 \tr\rho^2 \le \beta(\mu)^{-1}\tr\sum_{(x_1,x_2)\ne(x^\prime_1,x^\prime_2)}\mu(x_1,x_2)\rho(x_1,x_2)\mu(x^\prime_1,x^\prime_2)\rho(x^\prime_1,x^\prime_2).
\end{align*}
\end{claim}
}

\begin{proof}[Proof of Claim \ref{clm:purity upper bound}]
All sums over $y$ below are restricted to outputs satisfying $\Pr_{(X_1,X_2)\sim\mu}[F_n(X_1,X_2)=y]>0$.
For any two quantum states $\rho,\sigma$, 
the condition $\rho\sigma=0$ is necessary (and sufficient) for perfect discrimination of the two states 
(see, e.g., Proposition 5.13 in \cite{HIK+15}).
Since the referee $R$ perfectly discriminates 
$\rho(x_1,x_2)$ and $\rho(x^\prime_1,x^\prime_2)$
for which $F_n(x_1,x_2)\ne F_n(x^\prime_1,x^\prime_2)$ from the correctness, 
it must hold that $\rho(x_1,x_2)\rho(x^\prime_1,x^\prime_2)=0$.
Then, we have 
\begin{align*}
\lefteqn{ \tr\sum_{(x_1,x_2)\ne(x^\prime_1,x^\prime_2)}\mu(x_1,x_2)\rho(x_1,x_2)\mu(x^\prime_1,x^\prime_2)\rho(x^\prime_1,x^\prime_2)}\\
 &= \sum_y\tr\sum_{\substack{(x_1,x_2)\ne(x^\prime_1,x^\prime_2)\\ F_n(x_1,x_2)= F_n(x^\prime_1,x^\prime_2)=y}}\mu(x_1,x_2)\rho(x_1,x_2)\mu(x^\prime_1,x^\prime_2)\rho(x^\prime_1,x^\prime_2).
\end{align*}

From the privacy, there exists a quantum state $\rho_y$ for each $y$ such that 
$\rho_y = \rho(x_1,x_2)$ for every $(x_1,x_2)$ for which $F_n(x_1,x_2)=y$.
Therefore, 
\begin{align*}
 & \sum_y\tr\sum_{\substack{(x_1,x_2)\ne(x^\prime_1,x^\prime_2)\\ F_n(x_1,x_2)= F_n(x^\prime_1,x^\prime_2)=y}}\mu(x_1,x_2)\rho(x_1,x_2)\mu(x^\prime_1,x^\prime_2)\rho(x^\prime_1,x^\prime_2) \\
 & = \sum_y \tr\rho_y^2 \sum_{\substack{(x_1,x_2)\ne(x^\prime_1,x^\prime_2)\\ F_n(x_1,x_2)= F_n(x^\prime_1,x^\prime_2)=y}}\mu(x_1,x_2)\mu(x^\prime_1,x^\prime_2)\\
 & = 
 \sum_y \frac{\sum_{\substack{(x_1,x_2)\ne(x^\prime_1,x^\prime_2)\\ F_n(x_1,x_2)= F_n(x^\prime_1,x^\prime_2)=y}}\mu(x_1,x_2)\mu(x^\prime_1,x^\prime_2)}{\sum_{F_n(x_1,x_2)=F_n(x^\prime_1,x^\prime_2)=y}\mu(x_1,x_2)\mu(x^\prime_1,x^\prime_2)}
 \tr\rho_y^2 \!\!
 \sum_{F_n(x_1,x_2)=F_n(x^\prime_1,x^\prime_2)=y}\!\!\!\!\!\!\!\!\!\!\!\!\mu(x_1,x_2)\mu(x^\prime_1,x^\prime_2)\\
& = 
 \sum_y \CProb{}{(X_1,X_2)\ne (X^\prime_1,X^\prime_2)}{F_n(X_1,X_2)=F_n(X^\prime_1,X^\prime_2)=y} %edited by HN
 \tr\rho_y^2 \\
 & \ \ \ \ \ \ \ \ \times
 \sum_{F_n(x_1,x_2)=F_n(x^\prime_1,x^\prime_2)=y}\mu(x_1,x_2)\mu(x^\prime_1,x^\prime_2)\\
 & \ge \beta(\mu) \sum_y \tr\rho_y^2 
 \sum_{F_n(x_1,x_2)=F_n(x^\prime_1,x^\prime_2)=y}\mu(x_1,x_2)\mu(x^\prime_1,x^\prime_2)\\  
 &= \beta(\mu) \tr % \tr was added by HN 
 \sum_{x_1,x_2,x^\prime_1,x^\prime_2}\mu(x_1,x_2)\rho(x_1,x_2)\mu(x^\prime_1,x^\prime_2)\rho(x^\prime_1,x^\prime_2).
\end{align*}
\end{proof}
\ignore{
\begin{lemma}\label{lem:weight}
If $F_n$ is non-degenerate, 
we have 
\[
 \sum_{z\ne x_1} \lvert\braket{\psi_1(x_1;r)}{\psi_1(z;r^\prime)}\rvert^2 \le 1
\]
for every $r,r^\prime$ and every $x_1$.
Similarly
\[
 \sum_{z} \lvert\braket{\psi_2(x_2;r)}{\psi_2(z;r^\prime)}\rvert^2 \le 1
\]
for every $r,r^\prime$ and every $x_2$.
\end{lemma}
}

\begin{proof}[Proof of Lemma \ref{lem:weight}]
Fix $r$ and distinct $x,z\in\XX_1$.
By non-degeneracy, there exists $y\in\XX_2$ such that
$F_n(x,y)\ne F_n(z,y)$.
By perfect correctness, the referee distinguishes the joint messages
$\ket{\psi_1(x;r)}\otimes\ket{\psi_2(y;r)}$ and
$\ket{\psi_1(z;r)}\otimes\ket{\psi_2(y;r)}$ with certainty, and hence, they
are orthogonal. Since $\ket{\psi_2(y;r)}$ is a non-zero vector, this implies
$\braket{\psi_1(x;r)}{\psi_1(z;r)}=0$.
The same argument applies to $P_2$. Thus, for each $i\in\{1,2\}$ and
each fixed $r$, the family
$\{\ket{\psi_i(z;r)}:z\in\XX_i\}$ is orthonormal.

Applying Bessel's inequality to the orthonormal family for $r^\prime$
and the unit vector $\ket{\psi_i(x_i;r)}$, we obtain
\[
 \sum_{z\in\XX_i}
 \lvert\braket{\psi_i(x_i;r)}{\psi_i(z;r^\prime)}\rvert^2
 \le \norm{\ket{\psi_i(x_i;r)}}^2=1.
\]
Applying it to the orthonormal family for $r$ and the unit vector
$\ket{\psi_i(x_i;r^\prime)}$, we similarly obtain
\[
 \sum_{z\in\XX_i}
 \lvert\braket{\psi_i(z;r)}{\psi_i(x_i;r^\prime)}\rvert^2
 \le \norm{\ket{\psi_i(x_i;r^\prime)}}^2=1.
\]
\end{proof}

\end{document}